\documentclass{amsart}

\usepackage{amsmath}
\usepackage{amsthm}
\usepackage{amsfonts}
\usepackage{dsfont}

\newcommand{\reals}{\mathbb{R}}

\newcommand{\complex}{\mathbb{C}}

\newcommand{\integers}{\mathbb{Z}}

\newcommand{\angles}[1]{\left\langle #1 \right\rangle}

\newcommand{\norm}[1]{\left|\left|#1\right|\right|}

\newcommand{\para}[1]{\left(#1\right)}
\newcommand{\paraa}[1]{\big(#1\big)}
\newcommand{\parab}[1]{\Big(#1\Big)}
\newcommand{\parac}[1]{\bigg(#1\bigg)}
\newcommand{\parad}[1]{\Bigg(#1\Bigg)}

\newtheorem{theorem}{Theorem}[section]

\newtheorem{lemma}[theorem]{Lemma}
\newtheorem{proposition}[theorem]{Proposition}

\theoremstyle{definition}
\newtheorem{definition}[theorem]{Definition}
\theoremstyle{remark}

\numberwithin{equation}{section}

\renewcommand{\d}{\partial}
\newcommand{\h}{\hbar}
\renewcommand{\mid}{\mathds{1}}
\newcommand{\Ws}{W^*}
\newcommand{\Ls}{\Lambda^*}
\newcommand{\zb}{\bar{z}}
\newcommand{\qb}{\bar{q}}

\newcommand{\mv}{\vec{m}}
\newcommand{\nv}{\vec{n}}
\newcommand{\Tm}{T_{\mv}}

\newcommand{\Sn}{S_{\nv}}
\newcommand{\Zt}{\tilde{Z}}
\newcommand{\ximn}{\xi_{m,n}}
\newcommand{\Amutheta}{A_{\mu,\theta}}
\newcommand{\eps}{\varepsilon}
\renewcommand{\S}{\mathcal{S}}
\newcommand{\Chmu}{\mathcal{C}_{\hbar,\mu}}
\newcommand{\Ch}{\hat{C}}
\newcommand{\Xt}{\tilde{X}}
\newcommand{\Yt}{\tilde{Y}}
\renewcommand{\Zt}{\tilde{Z}}
\newcommand{\Wt}{\tilde{W}}
\newcommand{\thetat}{\tilde{\theta}}
\newcommand{\Azmut}{A^0_{\mu,\theta}}
\newcommand{\Us}{U^\ast}

\title[Deformed noncommutative tori]{Deformed noncommutative tori}

\author{Joakim Arnlind}
\address[Joakim Arnlind]{Department of Mathematics\\
Link\"oping University\\
581 83 Link\"oping\\
Sweden}
\email{joakim.arnlind@liu.se}

\author{Harald Grosse}
\address[Harald Grosse]{Mathematical Physics\\
Boltzmanngasse 5\\
1090 Vienna\\
Austria}
\email{harald.grosse@univie.ac.at}

\thanks{}

\subjclass[2000]{}
\keywords{}

\begin{document}

\begin{abstract}
  We recall a construction of non-commutative algebras related to a
  one-parameter family of (deformed) spheres and tori, and show that
  in the case of tori, the $\ast$-algebras can be completed into
  $C^\ast$-algebras isomorphic to the standard non-commutative
  torus. As the former was constructed in the context of matrix (or
  fuzzy) geometries, it provides an important link to the framework of
  non-commutative geometry, and opens up for a concrete way to study
  deformations of non-commutative tori.

  Furthermore, we show how the well-known fuzzy sphere and fuzzy torus
  can be obtained as formal scaling limits of finite-dimensional
  representations of the deformed algebras, and their projective modules are
  described together with connections of constant curvature.
\end{abstract}

\maketitle

\section{Introduction}

\noindent In \cite{abhhs:fuzzy,abhhs:noncommutative}, sequences of
matrix algebras where constructed as ``fuzzy'' analogues of surfaces
in $\reals^3$. In particular, a one parameter family of surfaces,
interpolating between spheres tori, was considered and all finite
dimensional (hermitian) representations of the corresponding
non-commutative algebras were found and classified. It was shown that
the representation theory clearly reflects the topology of the
surfaces and that smooth deformations of the geometry induce smooth
changes in the representations (see also \cite{a:repcalg,as:affine}).

One may wonder if there is a relation between the non-commutative
algebras constructed in this way and the $C^\ast$-algebra representing
the standard non-commutative torus \cite{c:cstardiff}? Although the
construction of the two algebras is completely different (and for the
deformed algebras, there is a priori no concept of norm), they
correspond to the same (non-commutative) topology. From this point of
view a $C^*$-algebraic isomorphism between the two kind of algebras
would be natural.

In this note, we review the construction of the non-commutative
algebras related to (deformed) spheres and tori, and show that there
exist formal scaling limits in which the representations become the
standard fuzzy sphere and torus. Furthermore, we construct a basis in
which one can complete the deformed torus algebras into
$C^\ast$-algebras isomorphic to the standard non-commutative torus.
Finally, we consider projective modules of the
non-commutative torus in this framework, together with
connections of constant curvature.

\section{Construction of  noncommutative algebras}

\noindent Let us briefly recall how to construct non-commutative
algebras related to level sets of a polynomial in $\reals^3$
\cite{abhhs:noncommutative,abhhs:fuzzy,a:repcalg,as:affine,a:phdthesis}. Given
a polynomial $C\in\reals[x,y,z]\equiv\reals[x^1,x^2,x^3]$, one can
define a Poisson bracket by setting
\begin{align}
  \{f,g\}=\nabla C\cdot\paraa{\nabla f\times\nabla g},
\end{align}
for $f,g\in C^\infty(\reals^3)$. In particular, it follows that
$\{x^i,x^j\}=\eps^{ijk}\d_kC$. By construction, the polynomial $C(x,y,z)$
Poisson commutes with all functions, which implies that the Poisson
structure restricts to the inverse image $C^{-1}(0)$. Thus, it defines
a Poisson structure on the quotient algebra
$\reals[x,y,z]/\angles{C(x,y,z)}$, which can be identified with
polynomial functions on $C^{-1}(0)$.

To construct a non-commutative version of the above algebra, one
starts with the free (non-commutative) associative algebra
$\complex[X,Y,Z]$ and imposes the relations
\begin{align*}
  &[X,Y] = i\hbar :\d_zC:\\
  &[Y,Z] = i\hbar :\d_xC:\\
  &[Z,X] = i\hbar :\d_yC:
\end{align*}
where $\hbar\in\reals$ and $:\d_i C:$ denotes a choice of ordering of
the (commutative) polynomial $\d_iC$.  In
\cite{abhhs:fuzzy,abhhs:noncommutative}, the authors considered the
polynomial
\begin{align}\label{eq:torusSpherePolynomial}
  C(x,y,z)=\frac{1}{2}\paraa{x^2+y^2-\mu}^2+\frac{1}{2}z^2-\frac{1}{2}
\end{align}
whose inverse image $C^{-1}(0)$ describes a sphere for $\mu<1$ and a
torus for $\mu>1$. One computes that 
\begin{align*}
  \{x,y\} = z\qquad
  \{y,z\} = 2x(x^2+y^2-\mu)\qquad
  \{z,x\} = 2y(x^2+y^2-\mu)
\end{align*}
and the corresponding non-commutative relations were chosen as
\begin{align}
  &[X,Y] = i\hbar Z\label{eq:XYCom}\\
  &[Y,Z] = i\hbar\parab{ 2X^3+XY^2+Y^2X-2\mu X}\label{eq:YZCom}\\
  &[Z,X] = i\hbar\parab{ 2Y^3+YX^2+X^2Y-2\mu Y}.\label{eq:ZXCom}
\end{align}
The algebra is then defined as $\Chmu = \complex[X,Y,Z]\slash I$,
where $I$ is the two-sided ideal generated by the above relations.
Using the ``Diamond Lemma'' \cite{b:diamondlemma}, it was proved that
$\Chmu$ is a non-trivial algebra for which a basis can be computed
\cite{abhhs:noncommutative}. We shall also consider $\Chmu$ to be a
$\ast$-algebra with $X^\ast=X$, $Y^\ast=Y$ and $Z^\ast=Z$.

In the Poisson algebra, the polynomial $C$ is a Poisson central
element of the algebra. It turns out that a non-commutative analogue
of $C$ is a central element in $\Chmu$. Namely, by setting
\begin{align}\label{eq:ncCentralElement}
  \Ch = \paraa{X^2+Y^2-2\mu\mid}^2+Z^2 
\end{align}
one computes that $[X,\Ch]=[Y,\Ch]=[Z,\Ch]=0$. Thus, in analogy with
(\ref{eq:torusSpherePolynomial}) it is natural to also impose
$\Ch=\mid$ in $\Chmu$. As we shall see, the
presentation of the algebra in terms of $X$, $Y$ and $Z$, is
appropriate when comparing with spherical geometries (and the ``fuzzy
sphere''); there is, however, another choice of basis which naturally
makes contact with non-commutative tori. By setting $W=X+iY$ and
eliminating $Z=\frac{1}{i\hbar}[X,Y]$, the remaining relations may be
written as
\begin{align}
  &\para{W^2\Ws+\Ws W^2}(1+\h^2)=4\mu\h^2W+2(1-\h^2)W\Ws W\label{eq:W}\\
  &\frac{1}{4}\paraa{W\Ws+\Ws W-2\mu\mid}^2+\frac{1}{4\h^4}\paraa{W\Ws-\Ws W}^2=\mid.
\end{align}
Furthermore, by introducing
\begin{align*}
  &\Lambda = \frac{1}{2\h}\paraa{W\Ws-\Ws W}+
  \frac{i}{2}\paraa{W\Ws+\Ws W-2\mu\mid}
\end{align*}
the above algebra can be presented as
\begin{align}
  &W\Lambda=q\Lambda W\quad;\quad \Ws\Lambda = \qb\Lambda \Ws\label{eq:arel1}\\
  &\Ws\Ls = q\Ls\Ws\quad;\quad W\Ls=\qb\Ls W\label{eq:arel2}\\
  &\Ls \Lambda=\Lambda\Ls = \mid\label{eq:arel3}\\
  &W\Ws = z\Lambda +\zb\Ls+\mu\mid\label{eq:arel4}\\
  &\Ws W = -\zb\Lambda - z\Ls + \mu\mid,\label{eq:arel5}
\end{align}
where $q=e^{2\pi i\theta}$, $z = e^{i\pi\theta}/2i\cos\pi\theta$ and
$\theta$ is related to $\hbar$ via $\hbar=\tan\pi\theta$.

Note that the above relations are quite similar to those of the
standard non-commutative torus, generated by two unitary
operators. The difference is the deformed unitarity of the operator
$W$. Let us now define the algebra together with the parameter ranges
that we shall be interested in.

\begin{definition}
  Let $\mu,\theta\in\reals$ such that $\mu>0$ and
  $|\mu\cos\pi\theta|>1$. By $\Azmut$ we denote the quotient of the
  (unital) free $\ast$-algebra $\complex\angles{W,\Ws,\Lambda,\Ls}$
  and the two-sided ideal generated by relations
  (\ref{eq:arel1})--(\ref{eq:arel5}).
\end{definition}

\noindent Again, one can make use of the Diamond lemma
\cite{b:diamondlemma} to explicitly compute a basis of $\Azmut$.

\begin{proposition}\label{prop:basis}
  A basis for $\Azmut$ is given by
  \begin{align*}
    &\Tm = q^{m_1m_2/2}\Lambda^{m_1}W^{m_2}\\
    &\Sn = q^{-n_1n_2/2}\Lambda^{n_1}(\Ws)^{n_2}
  \end{align*}
  where $\mv=(m_1,m_2)\in\integers\times\integers_{\geq 0}$ and
  $\nv=(n_1,n_2)\in\integers\times\integers_{\geq 1}$. Moreover, it
  holds that
  \begin{align*}
    &T_{\mv}T_{\nv} = q^{-\mv\times\nv/2}T_{\mv+\nv}\\
    &S_{\mv}S_{\nv} = q^{\mv\times\nv/2}S_{\mv+\nv}
  \end{align*}
  where $\mv\times\nv=m_1n_2-n_1m_2$.
\end{proposition}

\begin{proof}
  To prove that $\Tm$ and $\Sn$ provide a basis for the algebra, we
  make use of the ``Diamond Lemma'', and refer to
  \cite{b:diamondlemma} for details. Thus, relations
  (\ref{eq:arel1})--(\ref{eq:arel5}) are put into the reduction system
  \begin{align*}
    &S_1 = (W\Lambda,q\Lambda W)\quad
    S_2 = (W\Ls,\qb\Ls W)\quad
    S_3 = (\Ws\Ls,q\Ls\Ws)\\
    &S_4 = (\Ws\Lambda,\qb\Lambda\Ws)\quad
    S_5 = (\Lambda\Ls,\mid)\quad
    S_6 = (\Ls\Lambda,\mid)\\
    &S_7 = (W\Ws,z\Lambda+\zb\Ls+\mu\mid)\quad
    S_8 = (\Ws W,-\zb\Lambda-z\Ls+\mu\mid),
  \end{align*}
  and a compatible ordering is chosen as follows: if two words are of
  different total order (in $W,\Ws,\Lambda,\Ls$) then the one with
  lower order is smaller than the one with higher order. If two words
  are of the same order, they are comparable if the orders in
  $W,\Ws,\Lambda,\Ls$ are separately equal. Then the ordering is
  lexicographic with respect to the alphabet $\Lambda,\Ls,W,\Ws$.
  With this ordering one easily checks that $p_i\geq q_{ij}$, where
  $S_i=(p_i,\sum_{j}q_{ij})$. Furthermore, this ordering has the
  descending chain condition.

  There are several ambiguities to be checked in this reduction
  system. For instance, let us consider $W\Ws\Lambda$. One needs to
  check that it reduces to the same expression if we use $S_7$ to
  replace $W\Ws$ or $S_4$ to replace $\Ws\Lambda$. One computes
  \begin{align*}
    &\paraa{z\Lambda+\zb\Ls+\mu\mid}\Lambda-
    W\paraa{\qb\Lambda\Ws}=
    \paraa{z\Lambda+\zb\Ls+\mu\mid}\Lambda-q\qb\Lambda W\Ws\\
    &\qquad = \paraa{z\Lambda+\zb\Ls+\mu\mid}\Lambda
    -\Lambda\paraa{z\Lambda+\zb\Ls+\mu\mid} = 0.
  \end{align*}
  The other ambiguities can also be checked to be resolvable. Hence,
  by the Diamond Lemma, a basis for the algebra is provided by the
  irreducible words. Denoting $(\Ls)^n=\Lambda^{-n}$ the irreducible
  words are given by $\Tm$ and $\Sn$. To prove the product formulas,
  one simply uses the relations to reorder the expressions.
\end{proof}

\subsection{Scaling limits}

\noindent Let us now show that the algebras defined above have two
formal scaling limits reproducing the fuzzy sphere and the
standard non-commutative torus, as well as their finite dimensional
representations.

For the sphere, one introduces $\Xt=X/\eps$, $\Yt=Y\eps$, $\Zt=Z/\eps$
as well as $k=\hbar/\eps$. In terms of the rescaled variables,
relations (\ref{eq:XYCom})--~(\ref{eq:ZXCom}) become
\begin{align*}
  &[\Xt,\Yt] = ik\Zt\\
  &[\Yt,\Zt] = -2ik\mu\Xt+ik\eps^2\Phi(\Xt,\Yt)\\
  &[\Zt,\Xt] = -2ik\mu\Yt+ik\eps^2\Phi(\Yt,\Xt),
\end{align*}
where $\Phi(X,Y)=2X^3+XY^2+Y^2X$. For $-1<\mu<0$ one finds an algebra
isomorphic to su(2) as $\eps\to 0$. In \cite{abhhs:noncommutative},
the non-zero matrix elements of $W=X+iY$, in an $N$-dimensional irreducible
representation, were found to be
\begin{align*}
  W_{l,l+1} = \sqrt{\frac{2\sin(\pi l\theta)\sin\pi(n-l)\theta}{\cos\pi\theta}}
\end{align*}
for $l=1,\ldots,N-1$. Setting $\Wt=\Xt+i\Yt$ and $k=\tan\thetat$, one
finds that $\theta=\eps\thetat$ for small $\eps$, and
\begin{align*}
  \Wt_{l,l+1} \approx \frac{1}{\eps}\sqrt{2\pi^2\eps^2l(n-l)\thetat^2}
  \quad\to\quad \sqrt{2}\pi\thetat\sqrt{l(n-l)},
\end{align*}
which reproduces a standard representation of su(2).

To obtain the non-commutative torus, one considers relations
(\ref{eq:arel1})--~(\ref{eq:arel5}) and introduces $\Wt=\eps W$ and
sets $\mu=1/\eps^2$. Relations (\ref{eq:arel4}) and (\ref{eq:arel5})
then become
\begin{align*}
  &\Wt\Wt^\ast=\eps^2\paraa{z\Lambda+\zb\Ls}+\mid\\
  &\Wt^\ast\Wt=\eps^2\paraa{-z\Lambda-\zb\Ls}+\mid
\end{align*}
which reduces to the fact that $\Wt$ is unitary as $\eps\to
0$. Clearly, equations (\ref{eq:arel1})--(\ref{eq:arel3}) are
invariant under this rescaling. The corresponding non-zero matrix
elements are
\begin{align*}
  &W_{N,1} = \sqrt{\mu+\frac{1}{\cos\theta}}\\
  &W_{l,l+1} = \sqrt{\mu+\frac{\cos2\pi l\theta}{\cos\theta}}
\end{align*}
for $l=1,\ldots,N-1$, which implies that
\begin{align*}
  &\Wt_{N,1} = \eps\sqrt{\frac{1}{\eps^2}+\frac{1}{\cos\theta}}\quad\to\quad 1\\
  &\Wt_{l,l+1} = \eps\sqrt{\frac{1}{\eps^2}+\frac{\cos2\pi l\theta}{\cos\theta}}
  \quad\to\quad 1
\end{align*}
as $\eps\to 0$. Even without any rescaling, it holds that $\Lambda$ is
a matrix with the $N$ roots of unity on the diagonal. This reproduces
the finite dimensional representations of the non-commutative torus.

\section{Relation to the standard non-commutative torus} 

\noindent Let $U,V$ be the generators of the non-commutative torus
$C_\theta$ \cite{c:cstardiff,c:ncgbook,cr:yangmillstori}; i.e. the
universal $C^\ast$-algebra generated by the relations
\begin{align*}
  &VU = e^{i2\pi\theta}UV\\
  &U^\ast U=UU^\ast = \mid\\
  &V^\ast V=VV^\ast = \mid.
\end{align*}

\noindent In what follows, we will show that one can map $\Azmut$ into
$C_\theta$ and use the induced norm to complete $\Azmut$ to a
$C^\ast$-algebra isomorphic to $C_\theta$. Let us start by proving a
result about the spectrum of a particular element in $C_\theta$, that
is used to construct a $\ast$-homomorphism from $\Azmut$ to
$C_\theta$.

\begin{lemma}\label{lemma:invelement}
  If $|\mu\cos\pi\theta|>1$ and $\mu>0$ then the element
  $\mu\mid+ze^{i\pi\varphi}U+\zb e^{-i\pi\varphi}U^\ast$, with
  $z=e^{i\pi\theta}/2i\cos\pi\theta$, is positive and invertible in
  $C_\theta$ for all $\varphi\in\reals$.
\end{lemma}

\begin{proof}
  The element is clearly hermitian and let us write
  \begin{align*}
    \mu\mid+ze^{i\pi\varphi}U+\zb e^{-i\pi\varphi}U^\ast\equiv
    \mu\mid-B.
  \end{align*}
  To study the spectrum, we consider the invertibility of the element
  $(\mu-\lambda)\mid-B$ for different $\lambda$. It is a standard fact
  that this element is invertible if
  $\frac{1}{|\mu-\lambda|}\norm{B}<1$. One computes
  \begin{align*}
    \frac{1}{|\mu-\lambda|}\norm{B}
    =\frac{1}{2|(\mu-\lambda)\cos\pi\theta|}
    \norm{e^{i\pi\varphi}U+e^{i\pi\varphi}\Us}
    \leq \frac{1}{|(\mu-\lambda)\cos\pi\theta|},
  \end{align*}
  which is less than one if $|(\mu-\lambda)\cos\pi\theta|>1$. Since
  $|\mu\cos\pi\theta|>1$ by assumption (and $\mu>0$), it follows that
  $\frac{1}{|\mu-\lambda|}\norm{B}<1$ for all $\lambda\leq 0$. Hence,
  $\mu\mid-B$ is invertible and the spectrum is contained in
  $(0,\infty)$.
\end{proof}

\noindent Thus, it follows from Lemma \ref{lemma:invelement} that if $\mu>0$
is chosen such that $|\mu\cos\pi\theta|>1$ then both $\sqrt{\mu\mid+zU+\zb\Us}$
and its inverse exist in $C_\theta$.
\begin{proposition}
  The map $\phi$, defined by
  \begin{align*}
    &\phi(W) = \parab{\sqrt{\mu\mid+zU+\zb\Us}}V\\
    &\phi(\Lambda) = U,
  \end{align*}
  induces an injective $\ast$-homomorphism from $\Azmut$ to $C_\theta$.
\end{proposition}

\begin{proof}
  First of all, one has to check that the map is well defined,
  i.e. that it respects the relations in $\Azmut$; for instance,
  denoting $R=\mu\mid+zU+\zb\Us$, one computes
  \begin{align*}
    \phi(WW^\ast-z\Lambda-\zb\Ls-\mu\mid)
    &= \sqrt{R}VV^\ast\sqrt{R}-zU-\zb\Us-\mu\mid\\
    &= R - zU-\zb\Us-\mu\mid = 0,
  \end{align*}
  by using the relations in $C_\theta$. The remaining relations are
  checked in a similar way. Now, let us prove that $\phi$ is
  injective. It follows from Proposition \ref{prop:basis} that an
  arbitrary element $a\in\Azmut$ can be written as
  \begin{align*}
    a = \sum_{\mv\in\integers\times\integers_{\geq 0}}a_{\mv}\Lambda^{m_1}W^{m_2}
    +\sum_{\nv\in\integers\times\integers_{\geq 1}}b_{\nv}\Lambda^{n_1}(\Ws)^{n_2}
  \end{align*}
  (where all but a finite number of coefficients are zero) which
  implies that
  \begin{align*}
    \phi(a) = \sum_{\mv}a_{\mv}U^{m_1}(\sqrt{R}V)^{m_2}
    +\sum_{\nv}b_{\nv}U^{n_1}(V^\ast\sqrt{R})^{n_2}.
  \end{align*}
  Note that Lemma \ref{lemma:invelement} implies that
  $R(q^k)=\mu\mid+zq^kU+\zb q^k\Us$ is positive and invertible for all
  $k\in\integers$, which in particular implies that
  $(\sqrt{R})V=V\sqrt{R(\qb)}$ and $V^\ast\sqrt{R}
  = \paraa{\sqrt{R(\qb)}}V^\ast$. Thus, one concludes that
  \begin{align*}
    \phi(a) = \sum_{\mv}a_{\mv}\parad{\prod_{k=0}^{m_2-1}\sqrt{R(\qb^{k})}}U^{m_1}V^{m_2}
    +\sum_{\nv}b_{\nv}\parad{\prod_{k=1}^{n_2}\sqrt{R(\qb^k)}}U^{n_1}(V^\ast)^{n_2}.
  \end{align*}
  Since elements of the form
  \begin{align*}
    \sum_{\mv\in\integers\times\integers}c_{\mv}U^{m_1}V^{m_2}
  \end{align*}
  form a basis of a dense subset of $C_\theta$, it follows that if
  $\phi(a)=0$ then one must have $a_{\mv}=b_{\nv}=0$ for all $\mv$ and
  $\nv$. Hence, $a=0$, which proves that $\phi$ is injective.
\end{proof}

\noindent Since $\phi$ is injective, one can define a $C^\ast$-norm on
$\Azmut$ by setting $\norm{a}=\norm{\phi(a)}$ for all $a\in\Azmut$,
and by $\Amutheta$ we denote the completion of $\Azmut$ in this norm.
Moreover, $\phi$ can be extended to $\Amutheta$ by continuity, and (by
a slight abuse of notation) we shall also denote the extended map by
$\phi$.
\begin{proposition}
  The map $\phi:\Amutheta\to C_\theta$ is an isomorphism of $C^\ast$-algebras.
\end{proposition}

\begin{proof}
  As in Lemma \ref{lemma:invelement}, one can show that
  $\mu\mid+z\Lambda+\zb\Ls$ is positive and invertible in
  $\Amutheta$. Hence, one constructs the inverse of $\phi$ by
  setting
  \begin{align*}
    &\phi^{-1}(V) = \frac{1}{\sqrt{\mu\mid+z\Lambda+\zb\Ls}}W\\
    &\phi^{-1}(U) = \Lambda.
  \end{align*}
  (which is easily shown to be a well defined map) and extending it as a
  $\ast$-homomorphism through continuity.
\end{proof}

\subsection{Projective modules}

\noindent In \cite{abhhs:noncommutative} all finite-dimensional
hermitian $\ast$-rep\-resen\-tations of $\Chmu$ were constructed and
classified.  It was found that, in the case of algebras related to
tori, the parameter $\theta$ has to be a rational number for finite
dimensional representations to exist; which is in the same spirit as
for $C_\theta$. For the sake of comparison, let us see how the
standard projective modules of the non-commutative torus can be
presented for $\Amutheta$.

Let $\ximn$ be the vector space $\S(\reals\times\integers_n)$,
i.e. the space of Schwartz functions in one real variable $x$ and
one discrete variable $k\in\integers_n$. By defining
\begin{align}
  &\paraa{\phi W}(x,k) = W(x,k)\phi(x-\eps,k-1)\\
  &\paraa{\phi W^*}(x,k) = W(x+\eps,k+1)\phi(x+\eps,k+1)\\
  &\paraa{\phi\Lambda}(x,k)=e^{2\pi i(x-mk/n)}\phi(x,k)\\
  &\paraa{\phi\Ls}(x,k)=e^{-2\pi i(x-mk/n)}\phi(x,k),
\end{align}
where $\eps=(m+n\theta)/n$ and
\begin{align}
  W(k,x)=\parac{\mu+\frac{\sin\paraa{2\pi(x-mk/n)-\pi\theta}}{\cos\pi\theta}}^{\frac{1}{2}},
\end{align}
one can check that $\ximn$ becomes a right $\Azmut$ module.

The standard derivations on $C_\theta$, defined by
\begin{align*}
  &\d_1 U = iU \qquad \d_2 U = 0\\
  &\d_1 V = 0\qquad \d_2V = iV,
\end{align*}
and extended to the smooth part of $C_\theta$, can be pulled back to
the smooth part of $\Amutheta$ (defined as the inverse image of the
smooth part of $C_\theta$) giving
\begin{align*}
  &\d_1\Lambda = i\Lambda\qquad \d_2\Lambda = 0\\
  &\d_1W = i(z\Lambda-\zb\Ls)\paraa{\mu\mid+z\Lambda+\zb\Ls}^{-1}W\\
  &\d_2W = iW.
\end{align*}

\noindent Furthermore, a connection may be defined
on the above modules in a standard manner. Namely, the linear
operators $\nabla_1,\nabla_2:\ximn\to\ximn$, given as
\begin{align*}
  &\paraa{\nabla_1\phi}(x,k) = \frac{1}{2\pi}\frac{d\phi}{dx}(x,k)\\
  &\paraa{\nabla_2\phi}(x,k) = \frac{i}{\eps}x\phi(x,k),
\end{align*}
define a connection on $\ximn$, i.e. they fulfill
\begin{align*}\label{eq:connectionDef}
  \nabla_i(\phi\cdot a) = \paraa{\nabla_i\phi}\cdot a+\phi\cdot\paraa{\d_i a}
\end{align*}
for $i=1,2$ and $a$ in the smooth part of $\Amutheta$. One easily
computes that
\begin{align*}
  [\nabla_1,\nabla_2] = \frac{i}{2\pi\eps}\mid,
\end{align*}
i.e. the connection has constant curvature.

\bibliographystyle{alpha}
\bibliography{deformednctori}

\end{document}